\documentclass[11pt]{article}

\usepackage[margin = 2.50cm]{geometry}
\usepackage{microtype}
\usepackage{epsfig}
\usepackage{microtype}
\usepackage{wrapfig}
\usepackage{graphicx}
\usepackage{epsfig,amsmath,amsfonts,amssymb,amsthm,mathrsfs,ifpdf}
\usepackage{indentfirst,relsize}
\usepackage[numbers]{natbib}
\usepackage{setspace}
\usepackage{enumerate}
\usepackage{latexsym}
\usepackage{stackrel}
\usepackage[all]{xy}
\usepackage[usenames,dvipsnames]{pstricks}
\usepackage[linesnumbered,vlined,ruled]{algorithm2e}
\usepackage{pst-grad}
\usepackage{pst-plot}
\usepackage{xspace}
\usepackage{hyperref}

\usepackage{cleveref}

\newcommand{\size}[1]{\left| #1 \right|}

\newcommand{\E}{\mathbb{E}}
\newcommand{\remove}[1]{}

\newcommand{\R}{\mathbb{R}}
\newcommand{\N}{\mathbb{N}}

\newcommand{\cA}{\mathcal{A}}

\newcommand{\cC}{\mathcal{C}}

\newcommand{\cP}{\mathcal{P}}

\newcommand{\cG}{\mathcal{G}}

\newcommand{\Oh}{\mathcal{O}}
\newcommand{\cQ}{\mathcal{Q}}

\newcommand{\eps}{\epsilon}
\newcommand{\pr}{\mathbb{P}}

\theoremstyle{plain}
\newtheorem{theo}{Theorem}
\newtheorem{lem}[theo]{Lemma}

\theoremstyle{definition}
\newtheorem{defi}[theo]{Definition}
\newtheorem{rem}{Remark}
\newtheorem{obs}[theo]{Observation}

\newcommand{\iporacle}[1]{$\mbox{{\sc IP}}_{#1}$\xspace }

\newcommand{\mincut}{{\sc MinCut}\xspace}
\newcommand{\cutvalue}{{\mbox{\sc Cut}(G)} \xspace}
\newcommand{\cutvalueH}{{\mbox{\sc Cut}(H)} \xspace}
\newcommand{\mincutest}{{\sc MinCutEst}\xspace}
\newcommand{\mincutfind}{{\sc MinCutFind}\xspace}
\newcommand{\cutquery}{{\sc Cut Query}\xspace}
\newcommand{\demand}{{\sc Demand Query}\xspace}

\newcommand{\FindIntersection}{\textnormal{FIND-INT}}
\newcommand{\Intersection}{\text{INT}}
\newcommand{\tn}[1]{\textnormal{#1}}
\newcommand{\zone}{\{0,1\}}

\newcommand{\defproblem}[3]{
  \vspace{1mm}
\noindent\fbox{
  \begin{minipage}{0.96\textwidth}
  \begin{tabular*}{\textwidth}{@{\extracolsep{\fill}}lr} #1 \\ \end{tabular*}
  {\bf{Input:}} #2  \\
  {\bf{Output:}} #3
  \end{minipage}
  }
  \vspace{1mm}
}

\title{Query Complexity of Global Minimum Cut}

\author{
Arijit Bishnu
\and 
Arijit Ghosh
\and 
Gopinath Mishra
\and 
Manaswi Paraashar
}

\date{}

\begin{document}

\maketitle

\thispagestyle{empty}

\begin{abstract}
In this work, we resolve the query complexity of global minimum cut problem for a graph by designing a randomized algorithm for approximating the size of minimum cut in a graph, where the graph can be accessed through local queries like {\sc Degree}, {\sc Neighbor}, and {\sc Adjacency} queries.

Given $\epsilon \in (0,1)$, the algorithm with high probability outputs an estimate $\hat{t}$ satisfying the following
$(1-\epsilon) t \leq \hat{t} \leq (1+\epsilon) t$,
where $m$ is the number of edges in the graph and
$t$ is the size of minimum cut in the graph. The expected number of local queries used by our algorithm is $\min\left\{m+n,\frac{m}{t}\right\}\mbox{poly}\left(\log n,\frac{1}{\epsilon}\right)$ where $n$ is the number of vertices in the graph. Eden and Rosenbaum showed that $\Omega(m/t)$ many local queries are required for approximating the size of minimum cut in graphs. These two results together resolve the query complexity of the problem of estimating the size of minimum cut in graphs using local queries. 

Building on the lower bound of Eden and Rosenbaum, we show that, for all $t \in \mathbb{N}$, $\Omega(m)$ local queries are required to decide if the size of the minimum cut in the graph is $t$ or $t-2$. Also, we show that, for any $t \in \mathbb{N}$, $\Omega(m)$ local queries are required to find all the minimum cut edges even if it is promised that the input graph has
a minimum cut of size $t$. Both of our lower bound results are randomized, and hold even if we can make {\sc  Random Edge} query apart from local queries.

\end{abstract}


\section{Introduction}
\label{sec:intro}
\noindent
Global minimum cut (denoted \mincut) for a connected, unweighted, undirected and simple graph $G=(V,E)$, $\size{V}=n$ and $\size{E}=m$, is a partition of the vertex set $V$ into two sets $S$ and $V \setminus S$ such that the number of edges between $S$ and $V \setminus S$ is minimized. Let \cutvalue denote this edge set corresponding to a minimum cut in $G$, and $t$ denote $\size{\mbox{{\sc Cut}}(G)}$. The problem is so fundamental that researchers keep coming back to it again and again across different models~\cite{McGregor14,KargerS93,Karger93,KawarabayashiT19,MukhopadhyayN20,AhnGM12}. Fundamental graph parameter estimation problems, like estimation of the number of edges~\cite{Feige06,GoldreichR08}, triangles~\cite{EdenLRS15}, cliques~\cite{EdenRS18}, stars~\cite{GonenRS11}, etc.\ have been solved in the local and bounded query models~\cite{GoldreichGR98,GoldreichR08,KaufmanKR04}. 
Estimation of the size of \mincut is also in the league of such fundamental problems to be solved in the model of local queries. 

In property testing~\cite{Goldreich17}, a graph can be accessed at different granularities --- the query oracle can answer properties about graph that are local or global in nature. Local queries involve the relation of a vertex with its immediate neighborhood, whereas, global queries involve the relation between sets of vertices.
Recently using a global query, named \cutquery ~\cite{RubinsteinSW18}, the problem of estimating and finding \mincut was solved, but the problem of estimating or finding \mincut using local queries has not been solved. The fundamental contribution of our work is to resolve the query complexity of \mincut using local queries. We resolve both estimating and finding variants of the problem. To start with, we formally define the query oracle models we would be needing for discussions that follow. 

\paragraph*{The query oracle models.} We start with the most ubiquitous local queries and the random edge query for a graph $G=(V,E)$ where the vertex set $V$ is known but the edge set $E$ is unknown. 
\begin{itemize}
\item Local Query
\begin{itemize}
\item {\sc Degree} query: given $u \in V$, the oracle reports the degree of $u$ in $V$;
\item {\sc Neighbor} query: given $u \in V$, the oracle reports the $i$-th neighbor of $u$, if it exists; otherwise, the oracle reports $\perp$;
\item {\sc Adjacency} query:  given $u, \, v \in V$, the oracle reports whether $\{u,v\} \in E$.
\end{itemize}
 \item {\sc Random Edge} query: The query outputs an uniformly random edge of $G$.
\end{itemize}
\remove{
\paragraph*{The random edge query.} Random edge query for a graph $G$ is defined as
\begin{itemize}
    \item {\sc Random edge} query: The query outputs an uniformly random edge of $G$.
\end{itemize}
}
Apart from the above local queries, in the last few years, researchers have also used the {\sc Random Edge} query~\cite{AliakbarpourBGP18,AssadiKK19}. Notice that the randomness will be over the probability space of all edges, and hence, a random edge query is not a local query. We use this query in conjunction with local queries only for lower bound purposes. The other query oracle relevant for our discussion will be a \emph{global query} called the \cutquery proposed by Rubinstein et al.~\cite{RubinsteinSW18} that was motivated by submodular function minimization. The query takes as input a subset $S$ of the vertex set $V$ and returns the size of the cut between $S$ and $V \setminus S$ in the graph $G$.

\paragraph*{Prologue.}
Our motivation for this work is twofold --- \mincut is a fundamental graph estimation problem that needs to be solved in the local query oracle model and the lower bound of Eden and Rosenbaum~\cite{ER18} who extended the seminal work of Blais et al.~\cite{BlaisBM11} to develop a technique for proving query complexity lower bounds for graph properties via reductions from communication complexity. Using those techniques, for graphs that can be accessed by only local queries like {\sc Degree}, {\sc Neighbor}, {\sc Adjacency} and {\sc Random Edge}, 
Eden and Rosenbaum~\cite{ER18} showed that \mincut has a lower bound of $\Omega(m/t)$, where $m$ and $t$ are the number of edges and the size of the minimum cut, respectively, in the graph. In this work, we show that the query complexity of estimating \mincut using local queries only (and not {\sc Random Edge}) is $\widetilde{\Theta}\left( m/t \right)$, thus proving a matching upper bound. For designing the query algorithm for \mincut that matches the lower bound, we revisit the fundamental work of Karger~\cite{Karger93}. The power of query oracles allows us to use an ingenious coupling of a guessing scheme with Karger's result of upper bounding the number of cuts of a particular size, to come up with the algorithm. 

Prior to our work, no local query based algorithm has been developed for \mincut. But it was Rubinstein et al.\ ~\cite{RubinsteinSW18} who studied \mincut for the first time using \cutquery, a global query. They showed that there exists a randomized algorithm for finding a \mincut in $G$ using $\widetilde{\mathcal{O}}(n)$\footnote{$\widetilde{\mathcal{O}}(n)$ hides polylogarithmic terms in $n$.} many \cutquery.
Graur et al.~\cite{graur2019new} showed a matching (deterministic) lower bound for finding \mincut using \cutquery. 
\remove{
This leaves us in a piquant situation. We have local queries solving \mincut optimally and \cutquery, a global query, also solving it optimally with a deep chasm in the query complexities. Thus, there is a clear and huge separation in the power of the oracles. Is there any other query oracle to bridge the gap? We turn our attention to the \demand of Nisan~\cite{}, which is also a global query to solve \mincut.  
}

\remove{
\paragraph*{Demand query} A {\sc demand query}~\cite{} takes as input a vertex $v$ and an order on a set of vertices $U = (u_1, \ldots, u_n)$, $v \notin U$,  
and returns the index of the first vertex $i$ in the order such that there is an edge $(v, u_i)$ in the graph, or $0$ if none exists. 

A closely related oracle is the {\sc or query}~\cite{}, where the input is a vertex $v$ and a set $S$ of vertices, $(v \notin S)$, and 
returns an answer whether there exists an edge between $v$ and some
vertex in $S$. {\sc or query} is clearly weaker than the demand query, but binary search can be used to simulate demand queries using $O(\log n)$ {\sc OR queries}. 
} 

\paragraph*{Problem statements and results.}  \remove{Global minimum cut (\mincut) for a connected, unweighted, undirected graph $G=(V,E)$ is a partition of the vertex set $V$ into two sets $S$ and $V \setminus S$ such that the number of edges between $S$ and $V \setminus S$ is minimized.}  We focus on two problems in this work.
\remove{
\begin{description}
\item[Minimum Cut Estimation (\mincutest):] The input is a local query access to a graph $G$ and a parameter $\eps \in (0,1)$. The problem is to estimate $\size{\cutvalue}$, i.e., report a value of the cut that is within $(1+\eps) \cdot \size{\cutvalue}$.
\item[Minimum Cut Find (\mincutfind):] Here the problem is to explicitly find the set \cutvalue. \textcolor{blue}{finding $S$ is also hard.}
\end{description}
}

\defproblem{Minimum Cut Estimation}{A parameter $\eps \in (0,1)$, and access to an unknown graph $G$ via local queries}{An $(1 \pm \eps)$-approximation to $\size{\cutvalue}$.}\\

\defproblem{Minimum Cut Finding}{Access to an unknown graph $G$ via local queries}{Find a set \cutvalue.}

Our results are the following.

\begin{theo}{\bf (Minimum cut estimation using local queries)} 
\label{thm-mincut-estimation-local-queries}
There exists an algorithm, with {\sc Degree} and {\sc Neighbor} query access to an unknown graph $G = (V, E)$, that solves the minimum cut estimation problem with high probability. The expected number of queries used by the algorithm is 
$$
    \min\left\{m+n,\frac{m}{t}\right\}\mathrm{poly}\left(\log n,\frac{1}{\eps}\right).
$$
\end{theo}

Notice that Theorem~\ref{thm-mincut-estimation-local-queries} coupled with the matching lower bound result of Eden and Rosenbaum~\cite{ER18} closes the \mincut estimation problem in graphs using local queries.
\medskip

Building on the lower bound construction of Eden and Rosenbaum~\cite{ER18}, we show that no nontrivial query algorithm exists for finding a minimum cut or even estimating the exact size of a minimum cut in graphs.

\begin{theo}{\bf (Lower bound for minimum cut finding, i.e., \cutvalue)}
\label{theo-lower-bound-exact-cut}
   Let $m,n,t \in \N$ with $t \leq n-1$ and $ 2nt \leq m \leq {n \choose 2}$. Any algorithm that has access to {\sc Degree}, {\sc Neighbor}, {\sc Adjacency} and {\sc Random Edge} queries to an unknown graph $G = (V,E)$ must make at least $\Omega(m)$ queries in order to find all the edges in a minimum cut of $G$ with probability $2/3$.
   \remove{
   with $\size{V(G)} = n$, $\size{E(G)} = m$ and $\size{\mbox{\sc Cut}(G)} = t$, }
\end{theo}

\begin{theo}{\bf (Lower bound for finding the exact size of the minimum cut, i.e., $\size{\cutvalue}$)}
\label{theo-lb-count-exact-cut}
    Let $m,n,t \in \N$ with $2 \leq t \leq n-2$ and $ 2nt \leq m \leq {n \choose 2}$. Any algorithm that has access to {\sc Degree}, {\sc Neighbor}, {\sc Adjacency} and {\sc Random Edge} queries to an unknown graph $G = (V, E)$ must make at least $\Omega(m)$ queries in order to decide whether $\size{\mbox{\sc Cut}(G)}=t$ or $\size{\mbox{\sc Cut}(G)}=t-2$ with probability $2/3$. 
\end{theo}

 Local queries show a clear separation in its power in finding \mincut as opposed to the estimation problem. This is established by using the tight lower bound of minimum cut estimation (viz. $\Omega(m/t)$ lower bound of Eden and Rosenbaum and our Theorem~\ref{thm-mincut-estimation-local-queries}) vis-a-vis minimum cut finding as mentioned in our Theorems~\ref{theo-lower-bound-exact-cut} and ~\ref{theo-lb-count-exact-cut} on lower bound for finding \cutvalue. 
 \remove{
 Notice that Rubinstein et al.\ gave their upper bounds for finding \mincut and the same also holds for the estimation problem using \cutquery. Graur et al.'s~\cite{graur2019new} lower bound holds for finding \mincut.}

\paragraph*{Notations.}
In this paper, we denote the set $\{1,\ldots,n\}$ by $[n]$. For ease of notation we sometimes use $[n]$ to denote the set of vertices of a graph.
We say $x\geq 0$ is an $(1\pm \epsilon)$-approximation to $y\geq 0$ if $\size{x-y} \leq \epsilon y$. $V(G)$ and $E(G)$ would denote the vertex and edge sets when we want to make the graph $G$ explicit, else we use $V$ and $E$. For a graph $G$, $\mbox{\sc Cut}(G)$ denotes the set of edges in a minimum cut of $G$. Let $A_{1}$, $A_{2}$ be a partition of $V$, i.e., $V = A_{1} \cup A_{2}$ with $A_{1} \cap A_{2} = \emptyset$. Then, $\cC_G(A_{1},A_{2})= \left\{\{u,v\} \in E\; : \; u \in A_{1}\;\mbox{and}\; v\in A_{2}\right\}$.\remove{
Let $\Pi:\{0,1\}^N \times \{0,1\}^N \to \{0,1\}$. Then $R(\Pi)$ denotes the randomized communication complexity of $\Pi$.}
        The statement \emph{with high probability} means that the probability of success is at least $1-\frac{1}{n^c}$, where $c$ is a positive constant. 
        $\widetilde{\Theta}(\cdot)$ and $\widetilde{\Oh}(\cdot)$ hides a $\mbox{poly}\left(\log n, \frac{1}{\eps}\right)$ term in the upper bound.

\paragraph*{Organization of the paper}
Section~\ref{sec:queryalgo} discusses the query algorithm for estimating the \mincut while Section~\ref{sec:lowerbound} proves lower bounds on finding the \mincut. Section~\ref{sec-conclusion} concludes with a few observations.

\section{Estimation algorithm}
\label{sec:queryalgo}
\noindent
In this Section, we will prove Theorem~\ref{thm-mincut-estimation-local-queries}. In Section~\ref{sec:overview}, we talk about the intuitions and give the overview of our algorithm. We formalize the intuitions in Section~\ref{sec:algo-main}.

\subsection{Overview of our algorithm}
\label{sec:overview}
\noindent
We start by assuming that a lower bound $\hat{t}$ on $t=\size{\mbox{\sc Cut}(G)}$ is known. Later, we discuss how to remove this assumption.

We generate a random subgraph $H$ of $G$ by sampling each edge of the graph $G$  independently with probability $p=\Theta\left({\log n}/{\eps^2 \hat{t}}\right)$. Using Chernoff bound, we can show that any particular cut of size $k$, $k \geq t$, in $G$ is \emph{well approximated} in $H$ with probability at least $n^{-\Omega(k/\hat{t})}$. With this idea, consider the following Algorithm, stated informally, for minimum cut estimation.
\paragraph*{Algorithm-Sketch (works with $\hat{t} \leq t$)}
\begin{description}
\item[Step-1:] Generate a random subgraph $H$ of $G$ by sampling each edge in $G$  independently with probability $p={\Theta}\left({\log n}/{\eps^2 \hat{t}}\right)$. Note that $H$ can be generated by using $\widetilde{O}\left({m}/{\hat{t}}\right)$ many {\sc Degree} and {\sc Neighbor} queries in expectation. We will discuss it in Algorithm~\ref{algo:sample} in Section~\ref{sec:algo-main}.
\item[Step-2] Determine $\size{\cutvalueH}$ and report 
$\widetilde{t}=\frac{\size{\mbox{\sc Cut}(H)}}{p} $ as an $(1\pm \eps)$-approximation of $\size{\cutvalue}$.
\end{description}
The number of queries made by the above algorithm is $\widetilde{O}\left({m}/{\hat{t}}\right)$ in expectation. But it produces correct output only when the vertex partition corresponding to $\cutvalue$ and $\cutvalueH$ are the same. This is not the case always. If we can show that all cuts in $G$ is approximately preserved in $H$, then Algorithm-Sketch produces correct output with high probability. The main bottleneck to prove it is that the total number of cuts in $G$ can be exponential. A result of Karger (stated in the following lemma) will help us to make Algorithm-Sketch work. 
\begin{lem}[Karger~\cite{Karger93}]
\label{lem:karger_cut_bound}
For a given graph $G$ the number of cuts in $G$ of size at most $j \cdot \size{\mbox{\sc Cut}(G)}$ is at most $n^{2j}$.
\end{lem}
Using the above lemma along with Chernoff bound, we can show the following.
\begin{lem}
\label{lem:karger}
Let $G$ be a graph, $ \hat{t} \leq t = \size{\mbox{\sc Cut}(G)}$ and $\eps \in (0,1)$. If $H(V(G), E_p)$ be a subgraph of $G$ where each edge in $E(G)$ is included in $E_p$ with probability $p = \min \left\{\frac{200 \log n}{\eps^2 \hat{t}},1 \right\}$ independently, then every cut of size $k$ in $G$ has size $pk(1 \pm \eps)$ in $H$ with probability at least $1-\frac{1}{n^{10}}$.
\end{lem}
The above lemma implies the correctness of Algorithm-Sketch, which is for minimum cut estimation when we know a lower bound $\hat{t}$ of $\size{\cutvalue}$. But in general we do not know any such $\hat{t}$. To get around the problem, we start guessing $\hat{t}$ starting from $\frac{n}{2}$ each time reducing $\hat{t}$ by a factor of $2$. The guessing scheme gives the desired solution due to Lemma~\ref{lem:karger} coupled with the following intuition when $\hat{t}=\Omega(t\log n/\eps^2)$ --- if we generate a random subgraph $H$ of $G$ by sampling each edge with probability $p = \Theta\left({\log n}/{\eps^2 \hat{t}}\right)$, then $H$ is disconnected with at least a constant probability. So, it boils down to a connectivity check in $H$. The intuition is formalized in the following Lemma that can be proved using Markov's inequality.
\begin{lem}
Let $G$ be a graph with $\size{V(G)} = n$, $\hat{t} \geq \frac{2000 \log n}{\eps^2}\size{\mbox{\sc Cut}(G)}$ and $\eps \in (0,1)$. If $H(V(G), E_p)$ be a subgraph of $G$ where each edge in $E(G)$ is included in $E_p$ independently with probability $p=\min \left\{\frac{200\log n}{\eps^2 \hat{t}},1\right\}$, then $H$ is connected with probability at most $\frac{1}{10}$.
\label{lem:highguess1}
\end{lem}
\remove{
We cam make the probability $1/10$ to high probability by repeating the process stated in the above lemma $\widetilde{O}(1)$ times.
Now consider the following informal algorithm (Algorithm-2) for {Minimum Cut Estimation} that does not need a lower bound on $t$.
\paragraph*{Algorithm-2:}
\begin{description}
\item[Step-1:] Initialize $\hat{t}=n/2$,
\item[Step-2] Run Informal-Algorithm-1 for $\hat{t}$,
\item[Step-3:] Generate a random subgraph $H$ of $G$ by sampling each edge in $G$ with probability $p=\widetilde{\Theta}\left(\frac{1}{\hat{t}}\right)$.
\item[Step-4] If $H$ is disconnected, set $\hat{t}$ as $\frac{\hat{t}}{2}$ and go to Step-2. Otherwise, set $\hat{t}$ to roughly $\frac{\hat{t}}{\log n /eps^2}$ and go to Step-2.
\end{description}
 By Lemma~\ref{lem:highguess1}, we can say when the sampled subgraph $H$ is disconnected for all $\hat{t} \geq \Omega\left(\frac{\log n}{\eps^2t}\right)$ with high probability. The first $\hat{t}$, such that the sampled subgraph $H$ is connected, satisfies $\hat{t} =O\left(\frac{\log n}{\eps^2t}\right)$. That's why we set $\hat{t}$ to roughly $\frac{\hat{t}}{\log n /eps^2}$ and go to Step-2 and new $\hat{t}$ becomes less than $t$. By Lemma~\ref{lem:karger}, we have   }
Before moving to the next section, we prove Lemma~\ref{lem:karger} and~\ref{lem:highguess1} here.
\begin{proof}[Proof of Lemma~\ref{lem:karger}]
If $p=1$, we are done as the graph $H$ is exactly same as that of $G$. So, without loss of generality assume that the graph $G$ is connected. Otherwise, the lemma holds trivially as $\size{\mbox{\sc Cut}(G)} =0$, i.e., $\hat{t}=0$ and $p=1$. Hence, for the rest of the proof we will assume that $p=\frac{200\log n}{\eps^2 \hat{t}}$.

Consider a cut $\cC_G(A_{1}, A_{2})$ of size $k$ in $G$. As we are sampling each edge with probability $p$, the expected size of the cut $\cC_H(A_{1}, A_{2})$ is $pk$. Using Chernoff bound (see Lemma~\ref{lem:chernoff} in Section~\ref{sec:prob}), we get
\begin{equation}\label{eq:cher}
     \pr \left(\size{\mbox{$\cC_H(A_{1},A_{2})-pk$}} \geq \eps pk  \right) \leq e^{-\eps^2pk/3\hat{t}} 
             = n^{-\frac{100k}{3\hat{t}}}
\end{equation}

Note that here we want to show that every cut in $G$ is approximately preserved in $H$. To do so, we will use Lemma~\ref{lem:karger_cut_bound} along with Equation~\ref{eq:cher} as follows.  
Let $Z_1,Z_2,\ldots,Z_{\ell}$ be the partition of the set of all cuts in $G$ such that each cut in $Z_j$ has the number of edges between $\left[j\cdot \size{\mbox{\sc Cut}(G)},(j+1)\size{\mbox{\sc Cut}(G)}\right]$, where $\ell \leq \frac{n}{\size{\mbox{\sc Cut}(G)}}$ and $j \leq \ell-1$. From Lemma~\ref{lem:karger_cut_bound}, $\size{Z_j} \leq n^{2j}$. Consider a particular $Z_j, j \in [\ell]$. Using the union bound along with Equation~\ref{eq:cher}, the probability that there exists a cut in $Z_j$ that is not approximately preserved in $H$ is at most $\frac{1}{n^{11}}$. Taking union bound over all $Z_j$'s, the probability that there exists a cut in $G$ that is not approximately preserved is at most $\frac{1}{n^{10}}$.
\end{proof}
\begin{proof}[Proof of Lemma~\ref{lem:highguess1}]
Let $\cC_G(A_{1},A_{2})$ be a minimum cut in $G$.
Observe that 
$$
    \E \left[ \size{\cC_{H}(A_{1}, A_{2})}\right] = p\size{\cC_G(A_{1},A_{2})} = p \size{\mbox{\sc Cut}(G)}.
$$
The result follows from Markov's inequality.
$$ 
  \pr \left(\mbox{$G$ is connected}\right) \leq
 \pr \left(\size{\cC_G(A_{1}, A_{2} )} \geq 1\right) \leq \E[\size{\cC_G(A_{1}, A_{2})}] \leq \frac{1}{10}.
$$
\end{proof}

\subsection{Formal Algorithm (Proof of Theorem~\ref{thm-mincut-estimation-local-queries})}
\label{sec:algo-main}
\noindent
In this Section, the main algorithm for minimum cut estimation is described in Algorithm~\ref{algo:est} ({\sc Estimator}) that makes multiple calls to Algorithm~\ref{algo:guess} ({\sc Verify-Guess}). The {\sc Verify-Guess} subroutine in turn calls Algorithm~\ref{algo:sample} ({\sc Sample}) multiple times.

Given degree sequence of the graph $G$, that can be obtained using degree queries, we will first show how to independently sample each edge of $G$ with probability $p$ using only {\sc Neighbor} queries. 

\begin{algorithm}[H]\label{algo:sample}
\SetAlgoLined
\caption{{\sc Sample}(${ D},p)$}
\label{algo-sample}
\KwIn{${D}=\{d(i): i \in [n] \}$, where $d(i)$ denotes the degree of the $i$-th vertex in the graph $G$, and $p \in (0,1]$.}
\KwOut{Return a subgraph $H(V,E_p)$ of $G(V,E)$ where each edge in $E(G)$ is included in $E_p$ with probability $p$.}
\remove{Find $m=\frac{1}{2}\sum \limits_{i=1}^{d(i)}$.\\
\If{($p \geq \frac{100 \log n}{m}$)}
{
{\sc Find All the edges:} Extract the entire graph $G$ by using $D$ along with neighbor queries. Determine $\mbox{{\sc Cut }}_r(G)$ and return $G(V,E)$ as $H(V,E_p)$.
}}
Set $q=1-\sqrt{1-p}$ and $m=\frac{\sum_{i=1}^nd_i}{2}$ \\
\For{(each $i \in [n]$)}
{
\For{(each $j \in [d(i)]$ with $d(i) > 0$)}
{
// Let $r_{j}$ be the $j$-th neighbor of the $i$-th vertex \\
Add the edge $(i,r_{j})$ to the set $E_{p}$ with probability $q$; \\
}

}
Return the graph $H(V,E_p)$.
\end{algorithm}

The following lemma proves the correctness of the above algorithm {\sc Sample}(${ D},p)$.

\begin{lem}
\label{lem:sample_dp_lemma}
{\sc Sample}$( D,p)$ returns a random subgraph $H(V(G),E_p)$ of $G$ such that each edge $e \in E$ is included in $E_p$ independently with probability $p$. Moreover, in expectation, the number of {\sc Neighbor} queries made by {\sc Sample}$(D,p)$ is at most $2pm$.
\end{lem}
\begin{proof}
From the description of {\sc Sample}$(D, p)$, 
it is clear that 
the probability that a particular edge $e \in E(G)$ is added to $E_p$ with probability $1-(1-q)^2=p$.

Observe, $\E \left[\size{E_{p}} \right] = pm$. The bound on the number of {\sc Neighbor} queries now follows from the fact that {\sc Sample}$(D,p)$ makes at most $2\size{E_p}$ many {\sc Neighbor} queries.
\end{proof}

One of the core ideas behind the proof of Theorem~\ref{thm-mincut-estimation-local-queries} is that, given an estimate $\hat{t}$ of $t$, we want to efficiently (in terms of number of local queries used by the algorithm) decide if $\hat{t} \leq t$ or if $\hat{t}  \gtrsim \frac{\log n}{\eps^{2}} \times t$. Using Algorithm~\ref{algo}, we will show that this can be done using $\widetilde{\mathcal{O}}\left({m}/{\hat{t}}\right)$ many {\sc Neighbor} queries in expectation. Another interesting feature of Algorithm~\ref{algo} is that, if estimate $\hat{t} \leq t$, then Algorithm~\ref{algo} outputs an estimate which is a $(1\pm \eps)$-approximation of $t$.

\begin{algorithm}[h]\label{algo:guess}
\SetAlgoLined
\caption{{\sc Verify-Guess}(${ D},\hat{t},\eps$)}
\label{algo}
\KwIn{${ D}=\{d(i): i \in [n] \}$, where $d(i)$ denotes the degree of the $i$-th vertex in the graph $G$ and $m=\frac{1}{2}\sum\limits_{i=1}^{n}d(i) \geq n-1$. Also, a guess $\hat{t}$, with $1 \leq \hat{t} \leq \frac{n}{2}$, for the size of the global minimum cut in $G$, and $\eps \in (0,1)$.}
\KwOut{The algorithm should ``{\sc Accept}'' or ``{\sc Reject}'' $\hat{t}$, with high probability, depending on the following 
\begin{itemize}
    \item If $\hat{t} \leq \size{\mbox{\sc Cut}(G)}$, then {\sc Accept}
    $\hat{t}$ and also output an $(1\pm\eps)$-approximation of $\size{\mbox{\sc Cut}(G)}$ 
    
    \item If $\hat{t} \geq \frac{200\log n}{\eps^{2}} \size{\mbox{\sc Cut}(G)}$, then 
    {\sc Reject} $\hat{t}$
\end{itemize}
}
Set $p=\min \left\{\frac{200\log ^2 n}{\eps^2 \hat{t}},1\right\}$. \\
Set $\Gamma=100 \log n $ and Call {\sc Sample}($D,p$) $\Gamma$ times.\\
 Let $H_i(V,E^i_p)$ be the output of $i$-th call to {\sc Sample}($D,p$), where $i \in [\Gamma]$\\
\If{$\left(\mbox{at least}~{\Gamma}/{2} ~\mbox{many}~ H_i's\mbox{ are disconnected}\right)$}
{{\sc Reject} $\hat{t}$}
\ElseIf{(all $H_i$'s are connected)}
{{\sc Accept} $\hat{t}$, find $\mbox{{\sc Cut}}(H_i)$ for any $i \in [\Gamma]$, and return $\tilde{t}= \frac{ \size{\mbox{{\sc Cut}}(H_i)}}{p}$.}
\Else {
Return {\sc Fail}.\\
// When we cannot decide between ``{\sc Reject}'' or ``{\sc Accept}'' it will return {\sc Fail}
}
\end{algorithm}

The following lemma proves the correctness of Algorithm~\ref{algo}. The lemmas used in proof are Lemmas ~\ref{lem:karger}, ~\ref{lem:highguess1} and~\ref{lem:sample_dp_lemma}.

\begin{lem}
\label{lem:guess}
{\sc Verify-Guess}$(D,\hat{t},\eps)$ in expectation makes $\widetilde{\mathcal{O}}\left(\frac{m}{\hat{t}} \right)$ many {\sc Neighbor} queries to the graph $G$ and behaves as follows:
\begin{itemize}
    \item[(i)] If $\hat{t} \geq \frac{2000 \log n}{\eps^2}\size{\mbox{{\sc Cut}}(G)} $, then  {\sc Verify-Guess}$(D,\hat{t},\eps)$ rejects $\hat{t}$ with probability at least $1-\frac{1}{n^{9}}$.
    
    \item[(ii)] If $\hat{t} \leq \size{\mbox{{\sc Cut}}(G)} $, then {\sc Verify-Guess}$(D,\hat{t},\eps)$ accepts $\hat{t}$ with probability at least $1-\frac{1}{n^{9}}$. Moreover, in this case, {\sc Verify-Guess}$(D,\hat{t},\eps)$ reports an $(1 \pm \eps)$-approximation to $\mbox{{\sc Cut}}{(G)}$. 
\end{itemize}
\end{lem}
\begin{proof}
 {\sc Verify-Guess}$(D,\hat{t},\eps)$ calls {\sc Sample}($D,\eps$) for $\Gamma=100 \log n$ times with $p$ being set to $\min \left\{\frac{200\log n}{\eps^2 \hat{t}},1\right\}$. Recall, from Lemma~\ref{lem:sample_dp_lemma}, that each call to {\sc Sample}$(D,p)$ makes in expectation  at most $2pm$ many {\sc Neighbor} queries, and returns a random subgraph $H(V,E_p)$, where each edge in $E(G)$ is included in 
 $E_p$ with probability $p$. So, {\sc Verify-Guess}$(D,\hat{t},\eps)$ 
 makes  in expectation $\mathcal{O}(pm \log n)=\widetilde{\mathcal{O}}\left( {m}/{\hat{t}}\right)$ many 
{\sc neighbor} queries and generates $\Gamma$ many random subgraphs of $G$. The subgraphs are denoted by
 $H_1(V,E_p^1),\ldots,H_{\Gamma}(V,E_p^{\Gamma})$.

\begin{description}
\item[(i)] Let $\hat{t} \geq \frac{2000\log n}{\eps^2}\size{\mbox{{\sc Cut}}(G)}$. From Lemma~\ref{lem:highguess1}, we have that $H_i$ will be connected with probability at most $\frac{1}{10}$. Observe that in expectation, we get that at least $\frac{9\Gamma}{10}$ many $H_i$'s will be disconnected. By Chernoff bound (see Lemma~\ref{lem:chernoff} in Section~\ref{sec:prob}), the probability that at most $\frac{\Gamma}{2}$ many $H_i$'s are disconnected is at most $\frac{1}{n^{10}}$. Therefore, {\sc Verify-Guess}$(D,\hat{t},\eps)$ rejects any $\hat{t}$ satisfying $\hat{t} \geq \frac{2000 \log n}{\eps^2}\size{\mbox{{\sc Cut}}(G)}$ with probability at least $1-\frac{1}{n^{9}}$.

\item[(ii)]  Let $\hat{t} \leq \size{\mbox{{\sc Cut}}(G)}$. Using Lemma~\ref{lem:karger}, we have that every cut of size $k$ in $G$ has size $pk(1 \pm \eps)$ in $H_i$ with probability at least $1-\frac{1}{n^{10}}$. Therefore, with probability at least $1-\frac{\Gamma}{n^{10}}$, for all $i \in [\Gamma]$, every cut of size $k$ in $G$ has size $pk(1 \pm \eps)$ in $H_{i}$. This implies that if $\hat{t} \leq \size{\mbox{\sc Cut}(G)}$ then {\sc Verify-Guess}$(D,\hat{t},\eps)$ accepts any $\hat{t}$ with probability at least $1-\frac{1}{n^{9}}$. Moreover, for any $H_i$, observe that $\frac{\size{\mbox{{\sc Cut}}(H_i)}}{p}$ is an $(1 \pm \eps)$-approximation to $\size{\mbox{{\sc Cut}}(G)}$. Hence, when $\hat{t} \leq \size{\mbox{{\sc Cut}}(G)}$, {\sc Verify-Guess}$(D,\hat{t},\eps)$ also returns an $(1\pm \eps)$ approximation to $\size{\mbox{\sc Cut}(G)}$ with probability $1-\frac{1}{n^9}$.
\end{description}
\end{proof}

{\sc Estimator}($\eps$) (Algorithm~\ref{algo-estimator}) will estimate the size of the minimum cut in $G$ using {\sc Degree} and {\sc Neighbor} queries. The main subroutine used by the algorithm will be {\sc Verify-Guess}$(D,\hat{t},\eps)$.

\begin{algorithm}[h] \label{algo:est}
\SetAlgoLined
\caption{{\sc Estimator}($\eps$)}
\label{algo-estimator}
\KwIn{{\sc Degree} and {\sc Neighbor} query access to an unknown graph $G$, and a parameter $\eps \in (0,1)$.}
\KwOut{Either returns an $(1 \pm \eps)$-approximation to $t=\size{\mbox{{\sc Cut}}(G)}$ or {\sc Fail}}
Find the degrees of all the vertices in $G$ by making $n$ many {\sc degree} queries. Let $D=\{d(1),\ldots,d(n)\}$, where $d(i)$ denotes the degree of the $i$-th vertex in $G$.\\
If $\exists i \in [n]$ such that $d(i) = 0$, then return $t=0$ and {\sc Quit}. Otherwise, proceeds as follows.\\
Find $m=\frac{1}{2} \sum \limits _{i=1}^n d(i)$. If $m<n-1$, return $t=0$ and {\sc Quit}. Otherwise, proceed as follows.\\
Set $\kappa=\frac{2000 \log  n}{\eps^2}$\\
Initialize $\hat{t}=\frac{n}{2}$.\\
\While{($\hat{t} \geq 1$)}
{
Call {\sc Verify-Guess}$(D,\hat{t},\eps)$. \\
\If{({\sc Verify-Guess}$(D,\hat{t},\eps)$ returns {\sc Reject})}
{set $\hat{t}=\frac{\hat{t}}{2}$ and continue.}
\Else { 
// Note that in this case {\sc Verify-Guess}$(D,\hat{t},\eps)$ either returns {\sc Fail} or {\sc Accept}.\\
Set $\hat{t}_{u} =\max \left\{\frac{\hat{t}}{\kappa},1\right\}$.\\ 
Call {\sc Verify-Guess}($D,\hat{t}_{u},\eps$). \\
\If{({\sc Verify-Guess}($D,\hat{t}_{u},\eps$) returns {\sc Fail} or {\sc Reject})} 
{return {\sc Fail} as the output of {\sc Estimator}($\eps$)
}
\Else
{
Let $\tilde{t}$ be the output of {\sc Verify-Guess}($D,\hat{t}_{u},\eps$). \\
Return $\tilde{t}$ as the output of {{\sc Estimator}}$(\eps)$.
}
}}
{\bf Output:} Return that the graph $G$ is disconnected.
\end{algorithm}

The following lemma shows that with high probability {\sc Estimator}($\eps$) correctly estimates the size of the minimum cut in the graph $G$, and it also bounds the expected number of queries used by the algorithm.

\begin{lem}\label{thm:main}
{\sc Estimator}($\eps$) returns $(1 \pm \eps)$ approximation to $\size{\mbox{{\sc Cut}}(G)}$ with probability at least $1-\frac{1}{n^8}$ by making in expectation  $\min\left\{m+n,\frac{m}{t}\right\}\mbox{poly}\left(\log n,\frac{1}{\eps}\right)$ many queries and each query is either a {\sc Degree} or a {\sc Neighbor} query to the unknown graph $G$.
\end{lem}
\begin{proof}
Without loss of generality, assume that $n$ is a power of $2$. If $m < n-1$ or if there exists a $i \in [n]$ such that $d_{i} = 0$ then the graph $G$ is disconnected. In this case the algorithm {\sc Estimator}($\eps$) makes $n$ {\sc Degree} queries and returns the correct answer. Thus we assume that $m\geq n-1$. 

First, we prove the correctness and query complexity when the graph is
connected, that is, $t\geq 1$. Note that {\sc Estimator}($\eps$) calls
{\sc Verify-Guess}($D,\hat{t},\eps$) for different values of $\hat{t}$ starting from $\frac{n}{2}$. Recall that $\kappa=\frac{2000 \log n}{\eps^2}$. For a particular $\hat{t}$ with $\hat{t}\geq \kappa t$, {\sc Verify-Guess}($D,\hat{t},\eps$) does not {\sc Reject} $\hat{t}$ with probability at most $\frac{1}{n^{9}}$ by Lemma~\ref{lem:guess}~(i). So, by the union bound, the probability that {\sc Verify-Guess}($D,\hat{t},\eps$) will either {\sc Accept} or {\sc Fail} for some $\hat{t}$ with $\hat{t}\geq \kappa t$, is at most $\frac{\log n}{n^{9}}$. Hence, with probability at least
$1-\frac{\log n}{n^{9}}$, we can say that the {\sc Verify-Guess}($D,\hat{t},\eps$) rejects all $\hat{t}$ with $\hat{t}\geq \kappa t$. 

Observe that, from Lemma~\ref{lem:guess}~(ii), the first time $\hat{t}$ satisfy the following inequality
$$ 
    \frac{t}{2} <\hat{t} \leq t,
$$
{\sc Verify-Guess}($D,\hat{t},\eps$) will accept $\hat{t}$ with probability at least $1-\frac{1}{n^{9}}$.
Therefore, for the first time {\sc Verify-Guess}($D,\hat{t},\eps$) will either {\sc Accept} or {\sc Fail}, then $\hat{t}$  satisfies the following inequality
$$
    \frac{t}{2}< \hat{t} < \kappa t
$$
 with probability at least $1-\frac{\log n +1}{n^{9}}$. Let $\hat{t}_{0}$ denote the first time {\sc Verify-Guess} returns {\sc Accept} or {\sc Fail}. From the description of {\sc Estimator}($\eps$), note that, we get $\hat{t}_{u}$ by dividing $\hat{t}_{0}$ by $\kappa$. Note that, with probability at least $1-\frac{1+\log n}{n^{9}}$, we have $\hat{t}_{u} < t$. We then call the procedure {\sc Verify-Guess}($D,\hat{t},\eps$) with $\hat{t} = \hat{t}_{u}$. By
Lemma~\ref{lem:guess}~(ii), {\sc Verify-Guess}($D,\hat{t}_{u},\eps$) will {\sc Accept} and report an $(1\pm \eps)$ approximation to $t$ with probability at least $1-\frac{1}{n^{9}}$. 

We will now analyze the number
of {\sc Degree} and {\sc Neighbor} queries made by the algorithm. We make an initial $n$ many queries to construct the set $D$. Then at the worst case, we call {\sc Verify-Guess}($D,\hat{t},\eps$) for $\hat{t}=\frac{n}{2},\ldots
,t'$ and $\hat{t}=\frac{t'}{\kappa} \geq \frac{t}{2\kappa}$, where $ \frac{t}{2}< t' <\kappa t$. It is because 
{\sc Verify-Guess}($D,\hat{t},\eps$) accepts $\hat{t}$ with probability $1-\frac{1}{n^{9}}$ when the first time $\hat{t}$ satisfy the inequality $\hat{t}\leq t$. Hence, by Lemma~\ref{lem:guess} and the facts that $n \leq \frac{m}{t}$ and $\hat{t}_{u}\geq \frac{t}{2 \kappa}$ with probability at least $1-\frac{\log n + 1}{n^{9}}$, in expectation the total number of queries made by the algorithm is at most
$$
    n + \log n \cdot \left(1 -\frac{\log n + 1}{n^{9}}\right) \cdot \widetilde{\mathcal{O}}\left(\frac{2\kappa m}{t} \right) + \log n\cdot\left( \frac{\log n +1}{n^{9}}\right)\cdot \widetilde{\mathcal{O}}(m)= \widetilde{\mathcal{O}}\left( \frac{m}{t}\right).
$$
Note that each query made by {\sc Estimator}($\eps$) is either a {\sc Degree} or a {\sc Neighbor} query.

Now we analyze the case when $t=0$. Observe that {\sc Verify-Guess}($D,\hat{t},\eps$) rejects all $\hat{t} \geq 1$ with probability $1-\frac{\log n}{n^{9}}$, and therefore, {\sc Estimator}($\eps$) will report $t=0$. As we have called {\sc Verify-Guess}($D,\hat{t},\eps$) for all $\hat{t}=\frac{n}{2},\ldots,1$, the number of queries made by {\sc Estimator}($\eps$), in the case when $t=0$, is $\widetilde{\mathcal{O}}(m) + n$. Note that the additional $n$ term in the bound comes from the fact that to compute $D$ the algorithms needs to make $n$ many {\sc Degree} queries.  
\end{proof}

\section{Lower bounds}\label{sec:lowerbound}
\noindent
In this Section, we prove Theorems~\ref{theo-lower-bound-exact-cut} and~\ref{theo-lb-count-exact-cut} using reductions from suitable problems in communication complexity. In Section~\ref{sec:cc}, we discuss about two party communication complexity along with the problems that will be used in our reductions. We will discuss the proofs of Theorems~\ref{theo-lower-bound-exact-cut} and~\ref{theo-lb-count-exact-cut} in Section~\ref{sec:lbproof}.
\remove{
\begin{theo}{\bf (Lower bound for finding a minimum cut using local queries).}
\label{theo-lower-bound-exact-cut}
   Let $m,n,k \in \N$ with $t \leq n-1$ and $ 2nt \leq m \leq {n \choose 2}$. Any algorithm that has access to {\sc Degree}, {\sc Neighbor}, {\sc Adjacency} and {\sc Random edge} queries to an an unknown graph $G = (V(G), E(G))$  with $\size{V(G)} = n$, $\size{E(G)} = m$ and $\size{\mbox{\sc Cut}(G)} = t$, must make at least $\Omega(m)$ queries in order to find all the edges in a minimum cut of $G$ with probability $2/3$.
\end{theo}

\begin{theo}{\bf (Lower bound for finding the exact size of the minimum cut using local queries).}
\label{theo-lb-count-exact-cut}
    Let $m,n,k \in \N$ with $t \leq n-1$ and $ 2nt \leq m \leq {n \choose 2}$. Any algorithm that has access to {\sc Degree}, {\sc Neighbor}, {\sc Adjacency} and {\sc Random edge} queries to an an unknown graph $G = (V(G), E(G))$  with $\size{V(G)} = n$, $\size{E(G)} = m$ and $\size{\mbox{\sc Cut}(G)} \in \{t,t-2\} $, must make at least $\Omega(m)$ queries in order to decide whether $\size{\mbox{\sc Cut}(G)}=t$ or $\size{\mbox{\sc Cut}(G)}=t-2$ with probability $2/3$.
\end{theo}}

\remove{\begin{rem}
Note that when $m < 2nk$ then...
\end{rem}}

\subsection{Communication Complexity}\label{sec:cc}
\noindent
In two-party communication complexity there are two parties, Alice and Bob, that wish to compute a function $\Pi:\{0,1\}^N \times \{0,1\}^N \to \{0,1\} \cup \{0,1\}^n$~\footnote{The co-domain of $\Pi$ looks odd, as the the co-domain is $\{0,1\}$ usually. However, we need $\{0,1\}\cup \{0,1\}^n$ to take care of all the problems in communication complexity we discuss in this paper}. Alice is given ${\bf x} \in \{0,1\}^N$ and Bob is given ${\bf y} \in \{0,1\}^N$. Let $x_i~(y_i)$ denotes the $i$-th bit of ${\bf x}~({\bf y})$. While the parties know the function $\Pi$, Alice does not know ${\bf y}$, and similarly Bob does not know ${\bf x}$. Thus they communicate bits following a pre-decided protocol $\cP$ in order to compute $\Pi({\bf x},{\bf y})$. We say a randomized protocol $\cP$ computes $\Pi$ if for all $({\bf x},{\bf y}) \in \{0,1\}^N \times \{0,1\}^N$ we have $\pr[\cP({\bf x},{\bf y}) = \Pi({\bf x},{\bf y})] \geq 2/3$. The model provides the parties access to common random string of arbitrary length. The cost of the protocol $\cP$ is the maximum number of bits communicated, where maximum is over all inputs $({\bf x},{\bf y}) \in \{0,1\}^N \times \{0,1\}^N$. The communication complexity of the function is the cost of the most efficient protocol computing $\Pi$. For more details on communication complexity see~\cite{KN97}. We now define two functions {\sc $k$-Intersection} and {\sc Find-$k$-Intersection} and discuss their communication complexity. Both these functions will be used in our reductions.

\begin{defi}[{\sc Find-$k$-Intersection}]
\label{defi:find-intersection}
Let $k,N \in \N$ such that $k\leq N$.
Let $S=\{({\bf x},{\bf y}) \in \{0,1\}^N \times \{0,1\}^N: \sum_{i=1}^Nx_iy_i = k\}$. The {\sc Find-$k$-Intersection} function on $N$ bits is a partial function and is defined as $\FindIntersection_k^N: S \rightarrow \{0,1\}^N$, and is defined as $$\FindIntersection_k^N({\bf x},{\bf y}) = {\bf z}, \mbox{where}~ z_i=x_iy_i~ \mbox{for each}~ i \in [N].$$ 
\end{defi}
Note that the objective is that at the end of the protocol Alice and Bob know ${\bf z}$.

\begin{defi}[{\sc $k$-Intersection}]
\label{defi:k-intersection}
Let $k,N \in \N$ such that $k \leq N$. Let $S=\{({\bf x},{\bf y}):\sum\limits_{i=1}^N x_iy_i=k~\mbox{or}~ k-1\}$. The {\sc $k$-Intersection} function on $N$ bits is a partial function denoted by  $\tn{INT}_k^N:S   \rightarrow \{0,1\}$, and is defined as follows:
\begin{align*}
    \tn{INT}_k^N({\bf x},{\bf y}) = 
    \begin{cases}
    1 & \tn{ if } \sum_{i=1}^Nx_iy_i = k\\
    0 & \tn{otherwise}
    \end{cases}
\end{align*}
\end{defi}
In communication complexity, the {\sc $k$-Intersection} function on $N$ bits when $k=1$ is known as {\sc Disjointness} function on $N$.
\remove{\begin{theo}~\cite{KN97}
\label{theo:disj}
The randomized communication complexity of {\sc Disjointness} on $N$ bits is $\Omega(N)$.
\end{theo}}

\remove{Using standard reductions from {\sc Disjointness}, we can deduce the randomized communication complexities of {\sc $k$-Intersection} and {\sc Find-$k$-Intersection} as stated in Lemma~\ref{theo:kinter} and~\ref{theo:finding-intersection-hard}, respectively.}

\begin{lem}
\label{theo:finding-intersection-hard}
Let $k,N \in \N$ such that $k \leq cN$ for some constant $c<1$. The randomized communication complexity of {\sc Find-$k$-Intersection} function on $N$ bits is $\Omega(N)$.
\end{lem}

\begin{lem}
\label{theo:kinter}
Let $k,N \in \N$ such that $k \leq cN$ for some constant $c<1$. The randomized communication complexity of {\sc $k$-Intersection} function on $N$ bits ($\tn{INT}_k^N$) is $\Omega\left(N\right)$. 
 \end{lem}

\subsection{Proofs of Theorems~\ref{theo-lower-bound-exact-cut} and~\ref{theo-lb-count-exact-cut}}
\label{sec:lbproof}
\noindent
The proofs of Theorems~\ref{theo-lower-bound-exact-cut} and~\ref{theo-lb-count-exact-cut} are inspired from the lower bound proof of Eden and Rosenbaum~\cite{ER18} for estimating \mincut~\footnote{Note that Eden and Rosenbaum~\cite{ER18} stated the result in terms $k$-Edge Connectivity.}.

\begin{proof}[Proof of Theorem~\ref{theo-lower-bound-exact-cut}]
\label{proof-of-finding-lb}
We prove by giving a reduction from {\sc Find-$t/2$-Intersection} on $N$ bits. Without loss of generality assume that $t$ is even. Let ${\bf x}$ and ${\bf y}$ be the inputs of Alice and Bob. Note that $\sum\limits_{i=1}^N x_iy_i=t/2$.

We first discuss a graph $G_{({\bf x}{\bf y})}(V,E)$ that can be generated from $({\bf x},{\bf y})$, such that $\size{V}=n$ and $\size{E}=m\geq 2nt$, and works as the `hard' instance for our proof. Note that $G_{({\bf x},{\bf y})}$ should be such that no useful information about the \mincut can be derived by knowing only one of ${\bf x}$ and ${\bf y}$.  Let $s = t + \sqrt{t^2 + (m-nt)/2}$ and $N = s^2$. In particular, $2t \leq s \leq 2t + 3\sqrt{m}$. Also, $s \geq \sqrt{m/2}$ and therefore $s = \Theta(\sqrt{m})$.

\subsubsection*{{\bf The graph $G_{({\bf x},{\bf y})}$ and its properties:}}
$G_{({\bf x},{\bf y})}$  has the following structure.
\begin{itemize}
    \item $V = S_A \cup T_A \cup S_B \cup T_B \cup C$ such that $|S_A| = |T_A| = |S_B| = |T_B| = s$ and $|C| = n - 4s$. Let $S_A = \{s_i^A : i \in [s]\}$ and similarly  $T_A = \{t_i^A : i \in [s]\}$, $S_B = \{s_i^B : i \in [s]\}$ and $T_B = \{t_i^B : i \in [s]\}$.
    
    \item Each vertex in $C$ is connected to $2t$
    different vertices in $S_A$.
   
    \item For  $i,j \in [s]$: if $x_{ij} = y_{ij} = 1$, then $(s_i^A, t_j^B) \in E$ 
    and $(s_i^B, t_j^A) \in E$; otherwise, $(s_i^A, t_j^A) \in E$ and $(s_i^B, t_j^B) \in E$.
\end{itemize}
\begin{obs}
 $G_{({\bf x},{\bf y})}$ satisfies the following properties.
\begin{description}
\item[Property-1:] The degree of every vertex in $C$ is $2t$. For any $v \notin C$, the neighbors of $v$ inside $C$ are fixed irrespective of ${\bf x}$ and ${\bf y}$; and the number of neighbors outside $C$ is $s\geq 2t$.
    \item[Property-2:]  There are $t$ edges between the vertex sets $(C \cup S_A \cup T_A)$ and $(S_B \cup T_B)$, and removing them $G_{({\bf x},{\bf y})}$ becomes disconnected.
    
    \item[Property-3:] Every pair of vertices $(S_A \cup T_A \cup C)$ is connected by at least $3t/2$ edge disjoint paths. Also, every pair of vertices in $(S_B \cup T_B)$ is connected by at least $3t/2$ edge disjoint paths.
    
    \item[Property-4:] The set of $t$ edges between the vertex sets $(C \cup S_A \cup T_A)$ and $(S_B \cup T_B)$ forms the unique global minimum cut of $G({\bf x},{\bf y})$,
    
    \item[Property-5:] $x_{ij}=y_{ij}=1$ if and only if $(s_i^A,t_j^B)$ and $(s_i^B,t_j^A)$ are the edges in the unique global minimum cut of $G_{({\bf x},{\bf y})}$.
    
\end{description}
\end{obs}
\begin{proof}
Property-1 and Property-2 directly follow from the construction.  Now, we will prove Property-3. We first show that every pair of vertices $(S_A \cup T_A \cup C)$ is connected by at least $3t/2$ edge disjoint paths by breaking the analysis into the following cases. 
\begin{enumerate}
    \item[(i)] Consider $s_i^A, s_j^A \in S_A$, for $i,j \in [s]$. Under the promise that $\sum_{i = 1}^N x_iy_i = t/2$, $s_i^A, s_j^A$ have at least $s-t \geq 3t/2$ common neighbors in $T_A$ and thus there are at least $3t/2$ edge disjoint paths connecting them.
    
    \item[(ii)] Consider $s_i^A \in S_A$ and $t_j^A \in T_A$, for $i,j \in [s]$. Let $s_{j_1}^A, \dots, s_{j_{3t/2}}^A$ be $3t/2$ distinct neighbors of $t_j^A$ in $S_A$. Since, $s_i^A$ has $3t/2$ common neighbors with each $s_{j_r}^A$, $r \in [3t/2]$, there is a matching of size $3t/2$. Denote this matching by $(t_{j_r}^A, s_{j_r}^A)$, $r \in [3t/2]$. Thus $(s_i^A, t_{j_r}^A), (t_{j_r}^A, s_{j_r}^A), (s_{j_r}^A, t_{j}^A)$, for $r \in [3t/2]$, forms a set of edge disjoint paths of size $3t/2$ from $s_i^A$ to $t_j^A$, each of length $3$. In case $s_i^A$ is one of the neighbors of $t_j^A$, then one of the $3t/2$ paths gets reduced to $(s_i^A, t_{j}^A)$, a length $1$ path that is edge disjoint from the remaining paths.
    
    \item[(iii)] Consider $u, v \in C$. Let $u_1, \dots, u_{2t} \in S_A$ and $v_1, \dots, v_{2t} \in S_A$ be the neighbors of $u$ and $v$ respectively in $S_A$. If for some $i,j \in [2t]$, $u_i = v_j$ then $(u,u_i), (u_i,v_j), (v_j,v)$ is a desired path. Thus, assume $u_i \neq v_j$ for all $i,j \in [2t]$. For all $i \in [2t]$, since $u_i$ and $v_i$ have at least $3t/2$ common neighbors in $T_A$ we can find $3t/2$ edge disjoint paths $(u_i, t_i^A), (t_i^A, v_i)$, where $t_i^A \in T^A$. Existence of $3t/2$ edge disjoint paths from $u \in C$ to $v \in S_A$ can be proved as in (i). and from $u \in C$ to $v \in T_A$ can be proved as in (ii).
\end{enumerate}
Similarly, we can show that every pair of vertices in $(S_B \cup T_B)$ is connected by $3t/2$ many edge disjoint paths.

Observe that Property-4 follows from Property-3, and Property-5 
follows from the construction of $G_{({\bf x},{\bf y})}$ and Property-4.
\end{proof}
 
Now, by contradiction assume that there exists an algorithm $\cA$ that makes $o(m)$ queries to $G_{({\bf x},{\bf y})}$ and finds all the edges of a global minimum cut with probability $2/3$. Now, we give a protocol $\cP$ for {\sc Find-$t/2$-Intersection} on $N$ bits when the ${\bf x}$ and ${\bf y}$ are the inputs of Alice and Bob, respectively. Note that  $x,y \in \zone^N$ such that $\sum_{i = 1}^N x_iy_i = t/2$. 
\subsubsection*{Protocol $\cP$ for {\sc Find-$t/2$-Intersection}:}

Alice and Bob run the query algorithm $\cA$ when the unknown graph is $G_{({\bf x},{\bf y})}$. Now we explain how they simulate  the local queries and random edge query on $G_{({\bf x},{\bf y})}$ by communication. We would like to note that each query can be answered deterministically.
\begin{description}
    \item[{\sc Degree} query:] By Property-1, the degree of every vertex does not depend on the inputs of Alice and Bob, and therefore any degree query can be simulated without any communication.
    
    \item[{\sc Neighbor} query:] For $v \in C$, the  set of $2t$ neighbors are fixed by the construction. So, any neighbor query involving any $v \in C$ can be answered without any communication. For $i \in [s]$ and $s_i^A \in S_A$, let $N_C(s_i^A)$ be the set of fixed neighbors of $s_i^A$ inside $C$. So, by Property-1, $d(s_i^A)=\size{N_C(s_i^A)}+s$~\footnote{$d(u)$ denotes the degree of the vertex $u$ in $G_{({\bf x},{\bf y})}$}. The labels of the neighbors of $s_i^A$ are such that the first $\size{N_C(s_i^A)}$ many neighbors are inside $C$, and  they are arranged in a fixed but arbitrary order. For $j \in [s]$, the $(\size{N_C(v)}+j)$-th neighbor of $s_i^A$ is either $t_j^B$ or $s_j^A$ depending on whether $x_{ij}=y_{ij}=1$ or not, respectively. So, any neighbor query involving vertex in $S_A$ can be answered by $2$ bits of communication. Similar arguments also hold for the vertices in $S_B \cup T_A \cup T_B$.
    \item[{\sc Adjacency} query:] Observe that each adjacency query can be answered by at most $2$ bits of communication, and it can be argued like the {\sc Neighbor} query.
    \item[{\sc Random  Edge} query:] By Property-1, the degree of any vertex $v \in V$ is independent of the inputs of Alice and Bob. Alice and Bob use shared randomness to sample a vertex in $V$ proportional to its degree. Let $r \in V$ be the sampled vertex. They again use shared randomness to sample an integer $j$ in $[d(v)]$ uniformly at random. Then they determine the $j$-th neighbor of $r$ using {\sc Neighbor} query. Observe that this procedure simulates a {\sc Random Edge} query by using at most $2$ bits of communication.
\end{description}
Using the fact that $G_{({\bf x},{\bf y})}$ satisfies Property-$4$ and $5$, the output of algorithm $\cA$ determines the output of protocol $\cP$ for {\sc Find-$t/2$-Intersection}. As each query of $\cA$ can be simulated by at most two bits of communication by the protocol $\cP$, the number of bits communicated is $o(m)$. Recall that $N=s^2$ and $s=\Theta(\sqrt{m})$. So, the number of bits communicated by Alice and Bob in $\cP$ is  $o(N)$. This contradicts Theorem~\ref{theo:finding-intersection-hard}.
\end{proof}

\begin{proof}[Proof of Theorem~\ref{theo-lb-count-exact-cut}]
The proof of this theorem uses the same construction as the one used in the proof of Theorem~\ref{theo-lower-bound-exact-cut}. The `hard' communication problem to reduce from is {\sc $t/2$-Intersection} (see Definition~\ref{defi:k-intersection}) on $N$ bits, where $N=s^2$ and $s=\Theta(\sqrt{m})$.
\end{proof}

\section{Conclusion}
\label{sec-conclusion}

\paragraph{Global minimum $r$-way cut.}
Global minimum $r$-cut, for a graph $G=([n],E)$, $\size{V}=n$ and $\size{E}=m$, is a partition of the vertex set $[n]$ into $r$-sets $S_{1}, \, \dots, \, S_{r}$ such that the following is minimized
$$
    \left|\left\{\{i,j\}\in E \; : \; \exists k, \, \ell \, (k\neq \ell)\, \in [r],\;\mbox{with}\; i\in S_{k} \;\mbox{and}\; j \in S_{\ell}\right\}\right|.
$$
Let $\mbox{\sc Cut}_{r}(G)$ denote the set of edges corresponding to a minimum $r$-cut, i.e., the edges that goes across different partitions, and by the size of minimum $r$-cut, we mean $\size{\mbox{{\sc Cut}}_{r}(G)}$. The sampling and verification idea used in the proof of Theorem~\ref{thm-mincut-estimation-local-queries} can be extended directly, together with~\cite[Corollary~8.2]{Karger93}, to get the following result.

\begin{theo}
There exists an algorithm, with {\sc Degree} and {\sc Neighbor} query access to an unknown graph $G = ([n], E)$, that with high probability outputs a $(1\pm \eps)$-approximation of the size of the minimum $r$-cut of $G$. The expected number of queries used by the algorithm is 
$$
    \min\left\{m+n,\frac{m}{t_{r}}\right\}\mbox{poly}\left(r,\log n,\frac{1}{\eps}\right),
$$
where $t_{r} = \size{\mbox{\sc Cut}_{r}(G)}$.
\end{theo}

\paragraph{Minimum cuts in simple multigraphs.}
A graph with multiple edges between a pair of vertices in the graph but without any self loops are called {\em simple multigraphs}. If we have {\sc Degree} and {\sc Neighbor}\footnote{For simple multigraphs, we will assume that the neighbors of a vertex are stored with multiplicities.} query access to simple multigraphs then we can directly get the following generalization of Theorem~\ref{thm-mincut-estimation-local-queries}.
\begin{theo}{\bf (Minimum cut estimation in simple multigraphs using local queries)} 
\label{thm-multigraphs-mincut-estimation-local-queries}
There exists an algorithm, with {\sc Degree} and {\sc Neighbor} query access to an unknown simple multigraph $G = (V, E)$, that solves the minimum cut estimation problem with high probability. The expected number of queries used by the algorithm is 
$$
    \min\left\{m+n,\frac{m}{t}\right\}\mbox{poly}\left(\log n,\frac{1}{\eps}\right),
$$
where $n$ is the number of vertices in the multigraph, $m$ is the number of edges in the multigraph and $t$ is the number of edges in a minimum cut.
\end{theo}

\bibliographystyle{alpha}
\bibliography{reference.bib}

\appendix

\section{Probability Results}
\label{sec:prob}

\begin{lem} [See~\cite{Dubhashi09}]
\label{lem:chernoff}
Let $X = \sum_{i \in [n]} X_i$ where $X_i$, $i \in [n]$, are independent random variables,  $X_i \in [0,1]$ and $\mathbb{E}[X]$ is the expected value of $X$. Then
\begin{itemize}
\item[(i)] \label{ch01} For $\epsilon > 0$

	$\Pr [ |X - \mathbb{E}[X] | > \epsilon \mathbb{E} [X] ] \leq \exp{ \left(- \frac{\epsilon^2}{3} \mathbb{E}[X]\right)}.$
 
\remove{\item[(ii)] \label{ch02} if $t > 2 \mathbb{E}[X]$, then 
\begin{eqnarray}
	\Pr[ X > t ] \leq 2^{-t}
\end{eqnarray}}
\item[(ii)] \label{ch03} Suppose $\mu_L \leq \mathbb{E}[X] \leq \mu_H$, then for $0 < \epsilon < 1$
\begin{itemize}
\item[(a)] $\Pr[ X > (1+\epsilon)\mu_H] \leq \exp{\left( - \frac{\epsilon^2}{3} \mu_H\right)}$.
\item[(b)] $\Pr[ X < (1-\epsilon)\mu_L] \leq \exp{\left( - \frac{\epsilon^2}{2} \mu_L\right)}$.
\end{itemize}
\end{itemize}
\end{lem}

\end{document}